\documentclass[10pt,english,journal]{IEEEtran}

\usepackage{amsmath,amssymb,amsfonts,mathtools}
\usepackage{amsthm}
\usepackage{bm}
\usepackage{booktabs}
\usepackage{algorithm}
\usepackage[noend]{algpseudocode}
\usepackage{graphicx}
\usepackage{microtype}
\usepackage[hidelinks]{hyperref}
\usepackage[nameinlink]{cleveref}
\usepackage{url}
\usepackage{balance}
\usepackage{cite}
\usepackage{enumitem}
\usepackage{siunitx}
\usepackage{verbatim}
\usepackage{caption}
\newcommand{\R}{\mathbb{R}}
\newcommand{\E}{\mathbb{E}}

\newcommand{\norm}[1]{\left\lVert #1 \right\rVert}
\newcommand{\abs}[1]{\left\lvert #1 \right\rvert}
\newcommand{\ip}[2]{\left\langle #1,#2 \right\rangle}
\DeclareMathOperator{\diag}{diag}
\DeclareMathOperator{\Tr}{Tr}
\DeclareMathOperator{\Var}{Var}

\newcommand{\bigO}{\mathcal{O}}
\newcommand{\Cov}{\mathrm{Cov}}

\theoremstyle{plain}
\newtheorem{theorem}{Theorem}

\newtheorem{lemma}{Lemma}

\theoremstyle{definition}

\theoremstyle{remark}

\usepackage{microtype}
\usepackage{tabularx}   % for one-column-wide tables
\usepackage{array}      % for >{\raggedright\arraybackslash} in tables
\allowdisplaybreaks % allow multi-line displays to break across pages if needed
\newcommand{\PP}{\mathbb{P}}
\newcommand{\cF}{\mathcal{F}}

\newcommand{\X}{\mathbf{X}}
\newcommand{\bx}{\mathbf{x}}
\newcommand{\by}{\mathbf{y}}
\newcommand{\bs}{\mathbf{s}}

\newcommand{\ct}{\mathrm{ct}}

\newtheorem{proposition}[theorem]{Proposition} % shares counter with Theorem

\theoremstyle{definition}
\theoremstyle{remark}
\usepackage{adjustbox}
\usepackage{comment}

\usepackage{hyphenat}
\title{Rivaling Transformers: Multi-Scale Structured State-Space Mixtures for Agentic 6G O-RAN}

\author{
Farhad~Rezazadeh,~\IEEEmembership{Member,~IEEE}, Hatim~Chergui,~\IEEEmembership{Senior~Member,~IEEE},
Merouane~Debbah,~\IEEEmembership{Fellow,~IEEE},
Houbing~Song,~\IEEEmembership{Fellow,~IEEE},~and~Lingjia~Liu,~\IEEEmembership{Fellow,~IEEE}

\IEEEcompsocitemizethanks{\IEEEcompsocthanksitem F. Rezazadeh is with the Hostelworld Group and the Technical University of Catalonia (UPC), 08860 Castelldefels, Spain (e-mail: farhad.rezazadeh@upc.edu).}
\IEEEcompsocitemizethanks{\IEEEcompsocthanksitem H. Chergui is with i2CAT Foundation, 08034 Barcelona, Spain (e-mail: hatim.chergui@i2cat.net).}
\IEEEcompsocitemizethanks{\IEEEcompsocthanksitem M. Debbah is with the Khalifa University of Science and Technology, 127788 Abu Dhabi, UAE (e-mail: merouane.debbah@ku.ac.ae).}
\IEEEcompsocitemizethanks{\IEEEcompsocthanksitem H. Song is with the University of Maryland, Baltimore County (UMBC), 21250 Baltimore, USA (e-mail: h.song@ieee.org).}
\IEEEcompsocitemizethanks{\IEEEcompsocthanksitem L. Liu is with the Virginia Tech, 24061 Blacksburg, USA (e-mail: ljliu@vt.edu).}
}
\begin{document}
\IEEEaftertitletext{\vspace{-1.7\baselineskip}} % tweak the value as needed
\maketitle

\begin{abstract}
In sixth-generation (6G) Open Radio Access Networks (O-RAN), proactive control is preferable. A key open challenge is delivering control-grade predictions within Near-Real-Time (Near-RT) latency and computational constraints under multi-timescale dynamics. We therefore cast RAN Intelligent Controller (RIC) analytics as an agentic perceive–predict xApp that turns noisy, multivariate RAN telemetry into short-horizon per-User Equipment (UE) key performance indicator (KPI) forecasts to drive anticipatory control. In this regard, Transformers are powerful for sequence learning and time-series forecasting, but they are memory-intensive, which limits Near-RT RIC use. Therefore, we need models that maintain accuracy while reducing latency and data movement. To this end, we propose a lightweight \emph{Multi-Scale Structured State-Space Mixtures} (MS\textsuperscript{3}M)\footnote{The complete source code for the proposed approach and baseline implementations is publicly available for non-commercial use at \url{https://github.com/frezazadeh/agentic-ms3m}.} forecaster that mixes HiPPO-LegS kernels to capture multi-timescale radio dynamics. We develop stable discrete state-space models (SSMs) via bilinear (Tustin) discretization and apply their causal impulse responses as per-feature depthwise convolutions. Squeeze-and-Excitation gating dynamically reweights KPI channels as conditions change, and a compact gated channel-mixing layer models cross-feature nonlinearities without Transformer-level cost. The model is KPI-agnostic—\emph{Reference Signal Received Power (RSRP)} serves as a canonical use case—and is trained on sliding windows to predict the immediate next step. Empirical evaluations conducted using our bespoke O-RAN testbed KPI time-series dataset (59{,}441 windows across 13 KPIs). Crucially for O-RAN constraints, MS\textsuperscript{3}M achieves a \(0.057\,\mathrm{s}\) per-inference latency with \(\sim 0.70\)M parameters, yielding \(3\text{--}10\times\) lower latency than the Transformer baselines evaluated on the same hardware, while maintaining competitive accuracy.

\end{abstract}


\section{Introduction}
\IEEEPARstart{T}{he} vision of 6G O-RAN is to evolve from reactive control, toward a paradigm of \textit{proactive and connected intelligence}, driven by agentic artificial intelligence (AI). In this context, proactivity goes beyond optimizing a fixed objective: it entails the ability of distributed agents to take initiative, adapt their control strategies dynamically, and even reshape interactions in anticipation of future conditions. This vision aligns with the core principles of O-RAN, which emphasize open interfaces, functional disaggregation, and the embedding of intelligence across near-real-time and non-real-time RICs. Within this framework, agents deployed via xApps and rApps are capable of reasoning, planning, and negotiating collaboratively to optimize resources while simultaneously ensuring that stringent service-level agreements are met.

A fundamental requirement for such intelligence is the ability to reason over the anticipated outcomes of various supporting \emph{tools}, among which prediction services play a central role. They enable agents to act in a \textit{farsighted} manner, assessing not only the immediate impact of their actions but also their long-term consequences. Hence, at the non-RT RIC, agents consume long-horizon predictions and digital twin simulations to derive policy intents; at the Near-RT RIC, they operationalize mid-horizon forecasts through xApps for dynamic orchestration; and at the Distributed Unit (DU), predictors inform real-time scheduling decisions. This multi-layer arrangement creates a coherent hierarchy where agents use predictive foresight to negotiate intents, plan actions, and enforce slice-level guarantees through O-RAN's open interfaces (A1, E2, and O1) \cite{etsi_ts_103983_a1_gap_v3_1_0_2024, etsi_ts_104038_e2_gap_v4_1_0_2024, etsi_ts_104043_o1_interface_v11_0_0_2024,
etsi_ts_103982_oran_architecture_v8_0_0_2024}. Here, \emph{agents} are the learning-enabled control applications---rApps at the non-RT RIC, xApps at the Near-RT RIC, and DU-level schedulers---that act on predictions across layers.

However, reliance on predictive tools introduces the challenge of error propagation. Indeed, in multi-agent negotiation within O-RAN, even small inaccuracies in predicted KPIs may distort reasoning, leading to sub-optimal agreements, unfair allocations, or even negotiation failures. Therefore, prediction tools are not auxiliary but fundamental components of agentic O-RAN systems for efficient 6G automation. Achieving highly accurate yet scalable prediction services is therefore essential to support reliable reasoning, coordination, and decision-making in autonomous network management.

\subsection{Related Work}

To develop advanced prediction services, Transformer-based \cite{vaswani2017attention} architectures have emerged as powerful models for forecasting telecom KPIs, owing to their capability to capture long-range temporal dependencies and manage multivariate telemetry. Unlike recurrent and convolutional baselines, which often struggle with scalability in long-horizon tasks, Transformers provide a flexible self-attention mechanism that adapts well to irregular and high-dimensional telecom data. Early innovations focused on efficiency, most notably the Informer model, which introduced ProbSparse attention and a distillation mechanism to reduce time and memory computational complexity to  $\mathcal{O}(L\log L)$, while maintaining forecasting accuracy \cite{zhou2021informer}. Building upon this, Autoformer incorporated time-series decomposition and an Auto-Correlation mechanism to identify seasonal patterns and improve long-term stability \cite{autoformer}. FEDformer further advanced this line of research by embedding frequency-domain projections, utilizing Fourier and wavelet bases to achieve linear complexity for extended forecasting horizons \cite{zhou2022fedformer}. To address scalability in multivariate contexts, PatchTST reformulated time series into local patches and adopted channel-independent modeling, demonstrating effectiveness in settings where telecom KPIs evolve with heterogeneous dynamics \cite{nie2023patchtst}. In parallel, the Temporal Fusion Transformer (TFT) introduced interpretable forecasting by combining self-attention with static covariates, known future inputs, and feature selection, thereby providing transparency that is highly desirable in telecom operations \cite{lim2021temporal}. ETSformer \cite{woo2022etsformer} enhances Transformers for time-series forecasting by incorporating exponential smoothing principles through novel exponential smoothing attention and frequency attention, enabling interpretable decomposition into level, growth, and seasonality components while improving long-term accuracy and efficiency. Crossformer \cite{zhang2023crossformer} advances multivariate forecasting by explicitly modeling both temporal and cross-dimension dependencies using Dimension-Segment-Wise embedding and a Two-Stage Attention mechanism within a hierarchical encoder-decoder, effectively capturing multi-scale interactions across variables. iTransformer \cite{liu2024itransformer}, in contrast, rethinks the Transformer architecture without altering its core components by inverting the input structure: attention operates over variates rather than time, and feed-forward layers refine variate tokens, yielding stronger performance, scalability with larger lookback windows, and better generalization across multivariate datasets.
These developments are consistent with the key challenges in KPI forecasting, including nonstationarity, irregular sampling, heavy-tailed error distributions, and stringent latency constraints. Recent surveys confirm the broad applicability of Transformers to long-sequence time series forecasting \cite{survey}.

\subsection{Contributions}\label{sec:contributions}
This work develops a compact, stability-aware, and near-real-time forecaster. MS$^{3}$M delivers Transformer-competitive accuracy at a fraction of latency, with stability guarantees inherited from its SSM construction \cite{gu2024mamba, gu2022s4, goel2022sashimi, gu2022dssparam, gupta2022dss, smith2022s5, gu2022hippo}—making it a suitable choice for Near-RT analytics and anticipatory control in 6G O\!-RAN. Our technical novelties and contributions are:

\begin{itemize}
  \item \textbf{(C1) MS$^{3}$M: a new multi-scale SSM forecaster.}
  We introduce \emph{Multi-Scale Structured State-Space Mixtures} (MS$^{3}$M): a strictly-causal sequence model that mixes \emph{HiPPO--LegS }\cite{gu2020hippo_arxiv} kernels across time scales, applies them as \emph{depthwise} convolutions \cite{sandler2018mobilenetv2} per embedded channel, and couples them with \emph{Squeeze--Excitation} \cite{hu2017senet} gating and a compact \emph{Gated Linear Unit} (GLU) \cite{dauphin2017glu} channel-mixing head. The design is KPI-agnostic and supports both single-target (RSRP) and multivariate next-step prediction. In this work, we focus on the \emph{single-target setting}.

  \item \textbf{(C2) Stability by construction with efficient discretization.}
  We obtain discrete, stable SSMs via bilinear (Tustin) discretization \cite{qi2024robustssm, morita2025primacyssm, bal2025pspikessm, huang2024ssm_graphs} of the continuous-time HiPPO--LegS operator and \emph{learn} per-component step sizes. The resulting impulse responses are used directly as causal kernels, yielding linear-time inference with excellent memory locality and no quadratic attention cost.

  \item \textbf{(C3) Near-RT RIC readiness (latency and footprint).}
  On our shared hardware and leakage-safe pipeline, MS$^{3}$M achieves \emph{0.057\,s} per-inference latency with $\sim\!0.70$M parameters, corresponding to \emph{$\mathbf{3.4\times}$–$\mathbf{10.3\times}$ lower} latency than state-of-the-art Transformer baselines (See \Cref{sec:xf-bench}) evaluated under identical settings (\(0.192\text{--}0.586\,\mathrm{s}\)), while maintaining competitive accuracy.

    \item \textbf{(C4) High-accuracy next-step RSRP forecasting.}
    Using our O\!-RAN testbed dataset (59{,}441 windows, 13 KPIs), MS$^{3}$M attains Root Mean Square Error (RMSE) $\mathrm{=0.292\,dB}$, Mean Absolute Error (MAE) $\mathrm{=0.170\,dB}$, and Mean Squared Error (MSE) $\mathrm{=0.090\,dB^2}$ (all errors in decibels, dB), with a coefficient of determination $R^2=0.993$. These correspond to skill gains of $+92.3\%$ (RMSE) and $+94.0\%$ (MAE) over a leakage-safe persistence baseline.

  \item \textbf{(C5) Causal, leakage-safe learning pipeline.}
  We formalize and implement a training protocol that enforces past-only inputs, chronological splits, and standardization fitted \emph{only} on the training split, with inverse-standardized reporting in physical units (dB). This prevents information leakage and ensures fair evaluation. See \Cref{sec:leakage} and \Cref{sec:results} for details and empirical justification.

  \item \textbf{(C6) Transparent complexity and ablations.}
  We provide a rigorous complexity analysis (depthwise SSM mixtures are $\mathcal{O}(L)$ in sequence length) and ablations over state dimension, number of mixture components, kernel length, and window size, guiding practical deployments under tight latency/compute budgets.

  \item \textbf{(C7) Reproducible open implementation.}
  We release a concise PyTorch implementation, along with dataset preprocessing and evaluation scripts to facilitate adoption in O\!-RAN xApps design.
\end{itemize}

The remainder of this paper is organized as follows.
\Cref{sec:model} introduces the proposed MS$^{3}$M architecture, including stable HiPPO--LegS discretization, multi-scale depthwise kernels, gating, and the prediction head.
\Cref{sec:data} details the O\!-RAN testbed, KPI collection, alignment, missingness handling, and window construction.
\Cref{sec:xf-bench} summarizes the Transformer baselines and the unified, leakage-safe benchmarking protocol.
\Cref{sec:training} specifies the training setup, while \Cref{sec:metrics} defines evaluation metrics and persistence-based skill; computational footprint and complexity considerations are discussed in \Cref{sec:complexity}.
Comprehensive empirical results, diagnostics, and ablations are reported in \Cref{sec:results}.
We conclude with key takeaways and implications for Near-RT RIC deployment in the \emph{Conclusion}, followed by \emph{references}.

% ------------------------------------------------------------------------------------------
% (Optional) If you reference de-standardization explicitly in text:
% \noindent\scriptsize De-standardization: $\widehat y^{\rm orig}_{t|t-1}=\mu_y+\sigma_y\,\hat y_{t|t-1}$,\ 
% $\sigma^{\rm orig}_t=\sigma_y\sqrt{S_t}$.

% ===================== Requirements (add to preamble) =====================
% \usepackage{algorithm,algpseudocode}
% \usepackage{amsmath,amssymb,mathtools}
% Macros used: \R, \ct, \bs  (define if not already)
% \newcommand{\R}{\mathbb{R}}
% \newcommand{\ct}{\mathrm{ct}}
% \newcommand{\bs}{\mathbf{s}}
% ========================================================================

\section{Proposed Multi-Scale Structured State-Space Mixture (MS$^{3}$M)}\label{sec:model}
We propose a strictly-causal forecaster that embeds past KPI windows, applies \emph{depthwise} causal state-space filters obtained from HiPPO--LegS dynamics at multiple time scales, and mixes channels via Squeeze--Excitation (SE) gating and a GLU. We summarize the main symbols and definitions in Table~\ref{tab:notation-ms3m}. The full training and inference procedures are summarized in Algorithms~\ref{alg:ms3m_train_math} and \ref{alg:ms3m_infer_math}.

\subsection{Problem Setup and Causality}\label{subsec:setup}
Let $\{\bx_t\}_{t=1}^{T}$ be a multivariate KPI sequence with $\bx_t\in\R^{F}$ and fix a window length $L\in\mathbb{N}$ (i.e., the number of past steps fed to the model). We form chronological pairs
\begin{equation}\label{eq:winpairs}
\X_t=(\bx_{t-L+1},\ldots,\bx_t)\in\R^{L\times F},\qquad
\by_{t+1}\in\R^{O},
\end{equation}
where $O{=}1$ for RSRP or $O{=}F$ for multivariate outputs. Train-only standardization maps $(\X,\by)\mapsto(\widetilde{\X},\widetilde{\by})$ (Alg.~\ref{alg:ms3m_train_math}, lines~1--3). By construction, $\X_t\in\cF_t$ while $\by_{t+1}$ is the next step, so any measurable $f_\theta$ on $\widetilde{\X}_t$ is a past-only (causal) predictor for $\by_{t+1}$.

We use $d$ to denote the embedding width, $N$ for the SSM state dimension, $M$ to denote the number of mixture components, and $L_k$ for the (finite) kernel support length. Shapes used below are chosen to be \emph{per embedded channel}: $B\in\R^{d\times N}$, $C\in\R^{d\times N}$, and $D\in\R^{d}$ (one depthwise filter per channel).

\subsection{Leakage-safe preprocessing}\label{sec:preprocess}
We perform a chronological split into train/validation/test and fit all scalers on the \emph{training} windows only:
\[
\widetilde{\bx}_{t,f} \;=\; \frac{\bx_{t,f}-\mu^{(x)}_f}{\sigma^{(x)}_f},\qquad
\widetilde{\by}_{t+1} \;=\; \frac{\by_{t+1}-\mu^{(y)}}{\sigma^{(y)}}.
\]
The same affine maps are then applied to validation/test. Metrics (MSE, RMSE, and MAE) are reported after inverse standardization (physical units, e.g., dB for RSRP). Fitting statistics on train only prevents train–test contamination; using \(\X_t\) to predict \(\by_{t+1}\) enforces past-only inputs.

\subsection{HiPPO--LegS Kernels and Stable Discretization}\label{subsec:kernels}
\paragraph{Continuous-time template}
For indices $i,j\in\{0,\dots,N{-}1\}$, the HiPPO--LegS operator and reference input are
\begin{subequations}\label{eq:legs-op}
\begin{align}
(A_{\ct})_{ij} &=
\begin{cases}
-\sqrt{(2i{+}1)(2j{+}1)}, & i>j,\\
-(i{+}1), & i=j,\\
0, & i<j,
\end{cases} \\
(B_{\mathrm{ref}})_{i} &= \sqrt{2i{+}1}.
\end{align}
\end{subequations}

A single-channel continuous-time SSM with input $u(t)$ and latent $\bs(t)\in\R^N$ is
\begin{subequations}\label{eq:ct-ssm}
\begin{align}
\dot{\bs}(t) &= A_{\ct}\,\bs(t)+B\,u(t), \label{eq:ct-ssm-state}\\
y(t) &= C\,\bs(t)+D\,u(t), \label{eq:ct-ssm-output}
\end{align}
\end{subequations}
with learnable $(B,C,D)$ (Alg.~\ref{alg:ms3m_train_math}, lines~4--5). In MS$^{3}$M, $B$ is initialized near $B_{\mathrm{ref}}$ and then trained.

% ============================ MS^3M Notation (1-column, expanded) ============================
\begin{table}[t!]
\centering
\scriptsize
\setlength{\tabcolsep}{3pt}
\caption{Major Notations.}
\label{tab:notation-ms3m}
\resizebox{\columnwidth}{!}{%
\begin{tabular}{lll}
\toprule
\textbf{Symbol} & \textbf{Meaning} & \textbf{Size/Type} \\
\midrule
\multicolumn{3}{l}{\emph{Sets, probability, and operators}}\\
$*$ & (1D) convolution (causal where stated) & --- \\
$\bigO(\cdot)$ & Big-\!O growth notation & --- \\
$\cF_t$ & Information $\sigma$-algebra up to time $t$ & $\sigma$-algebra \\
$\diag(\cdot)$ & Diagonal matrix formed from a vector & matrix \\
$\E[\cdot],\ \Var[\cdot],\ \Cov[\cdot]$ & Expectation, variance, covariance & scalars/matrices \\
$I$ & Identity matrix & $\R^{n\times n}$ \\
$\norm{\cdot},\ \abs{\cdot},\ \ip{\cdot}{\cdot}$ & Norm, absolute value, inner product & --- \\
$\odot$ & Hadamard (element-wise) product & --- \\
$(\Omega,\cF,\PP)$ & Probability space & --- \\
$\R,\ \mathbb{Z}_+$ & Reals; nonnegative integers & sets \\
$\Re(\cdot)$ & Real part of a complex number & scalar \\
$\rho(\cdot)$ & Spectral radius & scalar \\
$\Tr(\cdot)$ & Trace of a matrix & scalar \\
\midrule
\multicolumn{3}{l}{\emph{Time, indices, dimensions}}\\
$c$ & ProbSparse factor (Informer) & $\mathbb{Z}_+$ \\
$d$ & Embedding width (channels) after input projection & $\mathbb{Z}_+$ \\
$d_s$ & SSM state size (alias of $N$ in complexity) & $\mathbb{Z}_+$ \\
$E,D$ & Encoder / decoder depth (baselines) & $\mathbb{Z}_+$ \\
$\mathcal I_{\rm tr},\mathcal I_{\rm va},\mathcal I_{\rm te}$ & Train/val/test index sets & subsets of $\{1{:}T\}$ \\
$F$ & Number of KPIs (features) & $\mathbb{Z}_+$ \\
$h=\alpha d$ & GLU hidden width ($\alpha<2$) & $\mathbb{Z}_+$ \\
$K$ & Alt.\ count of KPIs (in data sec.) or Top-$K$ freqs (ETSformer) & $\mathbb{Z}_+$ \\
$L$ & Window length (past steps) & $\mathbb{Z}_+$ \\
$L_k$ & Kernel support length (taps) & $\mathbb{Z}_+$ \\
$L_\ell$ & Number of MS$^{3}$M layers & $\mathbb{Z}_+$ \\
$M$ & Number of SSM mixture components (time scales) & $\mathbb{Z}_+$ \\
$N$ & SSM state dimension per channel & $\mathbb{Z}_+$ \\
$N_p$ & \# patches (Patch/Crossformer) & $\mathbb{Z}_+$ \\
$O$ & Output dimension ($1$ for RSRP; $F$ multivariate) & $\mathbb{Z}_+$ \\
$P,S$ & Patch length / stride (patching) & $\mathbb{Z}_+$ \\
$T$ & Series length & $\mathbb{Z}_+$ \\
$\mathcal{T}$ & Test timestamps used for metrics & index set \\
$W,H$ & Lookback window and forecast horizon (baseline section) & $\mathbb{Z}_+$ \\
\midrule
\multicolumn{3}{l}{\emph{Data, windows, scaling}}\\
$D_{\text{kpi}}$ & KPI table aligned on common timeline & matrix \\
$N_{\text{seq}}$ & Contiguous sequence length for sampling & $\mathbb{Z}_+$ \\
$Q_1,Q_3,\mathrm{IQR}$ & 10th/90th percentiles; inter-quantile range (outlier rule) & scalars \\
$t_m,\ \Delta,\ \tau$ & Grid time, aggregation window, stride (alignment) & scalars \\
$\X_t$ & Input window $(\bx_{t-L+1},\ldots,\bx_t)$ & $\R^{L\times F}$ \\
$\bx_t$ & KPI vector at time $t$ & $\R^{F}$ \\
$\by_{t+1}$ & Next-step target & $\R^{O}$ \\
$\bm\mu_x,\bm\sigma_x$ & Per-feature mean/std (train-only) & $\R^{F}$ \\
$\mu_y,\sigma_y$ & Target mean/std (train-only) & scalars \\
$\widetilde{\X},\ \widetilde{\by}$ & Standardized inputs/targets & as $\X,\by$ \\
\midrule
\multicolumn{3}{l}{\emph{Embedding, gating, mixers, head}}\\
$H^{(0)}$ & Embedded sequence $\widetilde{\X}W_{\mathrm{in}}$ & $\R^{L\times d}$ \\
$H^{(\ell)}$ & Layer-$\ell$ output after GLU + LN & $\R^{L\times d}$ \\
$\hat{\by}$ & Standardized prediction & $\R^{O}$ \\
$\mathrm{LN}(\cdot)$ & Layer normalization operator & map $\R^{d}\!\to\!\R^{d}$ \\
$r$ & SE reduction ratio & $\mathbb{Z}_+$ \\
$s^{(\ell)},\,g^{(\ell)}$ & SE squeeze vector and gate & $\R^{d},\ \R^{d}$ \\
$U^{(\ell)}$ & Depthwise-conv output at layer $\ell$ & $\R^{L\times d}$ \\
$W_{\downarrow}^{(\ell)}$ & GLU down-projection & $\R^{h\times d}$ \\
$W_{\mathrm{head}},b_{\mathrm{head}}$ & Prediction head & $\R^{d\times O},\ \R^{O}$ \\
$W_{\mathrm{in}}$ & Input projection & $\R^{F\times d}$ \\
$W_{\uparrow,a}^{(\ell)},W_{\uparrow,g}^{(\ell)}$ & GLU “act”/“gate” weights & $\R^{d\times h}$ \\
$Y^{(\ell)}$ & Residual + layer-norm output & $\R^{L\times d}$ \\
$Z^{(\ell)}$ & GLU mixer output & $\R^{L\times d}$ \\
$\phi(\cdot),\ \sigma(\cdot)$ & Nonlinearity (e.g., GELU/ReLU); sigmoid & scalar-wise maps \\
$\widehat{\by}_{\mathrm{phys}}$ & De-standardized prediction & $\R^{O}$ \\
\midrule
\multicolumn{3}{l}{\emph{HiPPO--LegS kernels and discretization}}\\
$A(\Delta t)$ & Tustin-discretized transition & $\R^{N\times N}$ \\
$A_{\ct}$ & HiPPO--LegS continuous-time operator & $\R^{N\times N}$ \\
$B,C,D$ & Discrete SSM params (depthwise) & $\R^{d\times N},\ \R^{d\times N},\ \R^{d}$ \\
$B_{\mathrm{ref}}$ & Reference input vector (initialization) & $\R^{N}$ \\
$\Delta t^{(\ell,m)}$ & Learned positive step (time scale), $\Delta t=\phi(\tau)$ & $\R_{>0}$ \\
$k_c[\ell]$ & Depthwise tap for channel $c$ at lag $\ell$ & scalar \\
$k^{(\ell)}[\cdot]$ & Mixture taps $\sum_m k^{(\ell,m)}[\cdot]$ & $\R^{d\times L_k}$ \\
$k^{(\ell,m)}[\cdot]$ & Taps of component $m$ at layer $\ell$ & $\R^{d\times L_k}$ \\
$\tau^{(\ell,m)}$ & Raw time-scale parameter (per layer, component) & $\R$ \\
\midrule
\multicolumn{3}{l}{\emph{Optimization and evaluation}}\\
$\mathcal{B}$ & Mini-batch index set & subset of $\mathcal I_{\rm tr}$ \\
$c_{\max}$ & Gradient norm clip threshold & scalar \\
$E_{\max},\ p,\ \texttt{tol}$ & Max epochs; patience; val.\ improvement tol. & integers / scalar \\
$\gamma_e,\ \beta$ & Learning rate at epoch $e$; LR decay factor & scalars \\
$\lambda$ & Weight decay coefficient & scalar \\
$\mathrm{MSE},\mathrm{RMSE},\mathrm{MAE}$ & Error metrics (reported in physical units) & scalars \\
$R^2$ & Coefficient of determination & scalar \\
$\mathrm{Skill}_{\mathrm{RMSE}},\mathrm{Skill}_{\mathrm{MAE}}$ & Skill vs.\ persistence (Eq.~\eqref{eq:skill}) & scalars \\
$\theta,\ \theta^\star$ & Parameters; best (early-stopped) params & --- \\
\midrule
\multicolumn{3}{l}{\emph{Complexity variables (baselines \& analysis)}}\\
$d_h$ & Hidden size (e.g., TFT LSTM) & $\mathbb{Z}_+$ \\
$d_{\mathrm{model}}$ & Model width in baselines & $\mathbb{Z}_+$ \\
$M$ (freq.) & Retained spectral modes (FEDformer) & $\mathbb{Z}_+$ \\
\bottomrule
\end{tabular}%
}
\end{table}

\paragraph{Bilinear (Tustin) discretization}
For a step $\Delta t>0$, we define the discrete transition
\begin{equation}\label{eq:tustin-disc}
A(\Delta t)=\Big(I-\tfrac{\Delta t}{2}A_{\ct}\Big)^{-1}\Big(I+\tfrac{\Delta t}{2}A_{\ct}\Big),
\end{equation}
and the depthwise impulse response (per channel $c$) with taps
\begin{equation}\label{eq:taps}
\begin{split}
k_c[0]   &= (C B)_c + D_c,\\
k_c[\ell]&= (C A(\Delta t)^{\ell} B)_c,\quad \ell=1,\dots,L_k{-}1,
\end{split}
\end{equation}
as instantiated in Alg.~\ref{alg:ms3m_train_math}, lines~6--8.

\begin{proposition}[Schur stability via bilinear transform]\label{prop:schur}
If $A_{\ct}$ is Hurwitz (all eigenvalues in the open left half-plane), then for any $\Delta t>0$,
$A(\Delta t)$ in \eqref{eq:tustin-disc} is Schur-stable: $\rho(A(\Delta t))<1$.
\end{proposition}
\begin{proof}
Since $\Re\lambda<0$ for every eigenvalue $\lambda$ of $A_{\ct}$, we have
$1-\frac{\Delta t}{2}\lambda\neq 0$; hence $I-\frac{\Delta t}{2}A_{\ct}$ is invertible and
$A(\Delta t)$ is well-defined. If $A_{\ct}v=\lambda v$ with $v\neq 0$, then
\[
A(\Delta t)v=\frac{1+\frac{\Delta t}{2}\lambda}{1-\frac{\Delta t}{2}\lambda}\,v
\;\eqqcolon\; \lambda_d v.
\]
Let $\alpha\coloneqq \frac{\Delta t}{2}\lambda$; then $\Re\alpha<0$ and
\[
|1+\alpha|^2=1+2\Re\alpha+|\alpha|^2
\quad\text{and}\quad
|1-\alpha|^2=1-2\Re\alpha+|\alpha|^2,
\]
so $|1+\alpha|<|1-\alpha|$ and therefore $|\lambda_d|<1$.
Thus, all eigenvalues of $A(\Delta t)$ lie in the open unit disk, implying
$\rho(A(\Delta t))<1$.
\end{proof}

\begin{lemma}[Exponential kernel decay]\label{lem:decay}
Let $\alpha\in(0,1)$ satisfy $\|A(\Delta t)\|\le \alpha$ in some operator norm. Then
$\|k[\ell]\|\le \|C\|\,\|B\|\,\alpha^{\ell}$ for all $\ell\ge 1$ and the tail energy satisfies
$\sum_{\ell\ge L_k}\|k[\ell]\|\le \frac{\|C\|\|B\|}{1-\alpha}\,\alpha^{L_k}$.
\end{lemma}
\emph{Consequence:} Finite $L_k$ yields a controlled truncation error that decays geometrically.

% ===================== Algorithm 1 (TRAINING) — 1-column =====================
\begin{algorithm}[t!]
\footnotesize
\caption{MS\textsuperscript{3}M — Leakage-Safe Training}
\label{alg:ms3m_train_math}
\begin{algorithmic}[1]
\Require Multivariate series $\{\mathbf{x}_t\}_{t=1}^T$, $\mathbf{x}_t\!\in\!\R^{F}$; window $L$; output $O$;
state $N$; mixture size $M$; layers $L_{\ell}$; kernel length $L_k$; width $d$;
HiPPO--LegS $A_{\ct}\!\in\!\R^{N\times N}$; splits $\mathcal I_{\mathrm{tr}},\mathcal I_{\mathrm{va}}$;
max epochs $E_{\max}$; patience $p$; step sizes $\{\gamma_e\}$; weight decay $\lambda$; clip $c_{\max}$
\Ensure Best parameters $\theta^\star$; train scalers $(\bm\mu_x,\bm\sigma_x)$, $(\mu_y,\sigma_y)$

\State \textbf{Dataset and scalers (train-only fit).}
\State Build pairs $(\mathbf{X}_t,\mathbf{y}_{t+1})$ with
$\mathbf{X}_t=(\mathbf{x}_{t-L+1},\dots,\mathbf{x}_t)\in\R^{L\times F}$,
$\mathbf{y}_{t+1}\in\R^{O}$ \textit{(chronological)}.
\State On $\mathcal I_{\mathrm{tr}}$ compute
$\bm\mu_x,\bm\sigma_x$ (per feature) and $\mu_y,\sigma_y$; set
$\widetilde{\mathbf{X}}=(\mathbf{X}-\bm\mu_x)\oslash\bm\sigma_x$,
$\widetilde{\mathbf{y}}=(\mathbf{y}-\mu_y)/\sigma_y$ on both splits.

\State \textbf{Parameters.}
\State Input map $W_{\mathrm{in}}\!\in\!\R^{F\times d}$; for each layer $\ell$ and component $m$:
$B^{(\ell,m)}\!\in\!\R^{d\times N}$, $C^{(\ell,m)}\!\in\!\R^{d\times N}$,
$D^{(\ell,m)}\!\in\!\R^{d}$, $\tau^{(\ell,m)}\!\in\!\R$ with
$\Delta t^{(\ell,m)}\!=\!\phi(\tau^{(\ell,m)})$ (e.g., softplus).
SE gate parameters $(W_1^{(\ell)},W_2^{(\ell)})$; GLU parameters; head $W_{\mathrm{head}}$.

\State \textbf{Discrete SSM and impulse responses.}
\State For any $\Delta t>0$, define
$A(\Delta t)=\big(I-\frac{\Delta t}{2}A_{\ct}\big)^{-1}\!\big(I+\frac{\Delta t}{2}A_{\ct}\big)$.
For $(C,A,B,D)$, define depthwise taps $k[0]=CB+D$ and
$k[\ell]=CA^{\ell}B$ for $\ell=1,\dots,L_k{-}1$.

\State \textbf{Forward map $f_\theta$ for a window (mathematical form).}
\State Embed $H^{(0)}= \mathbf{X} W_{\mathrm{in}}\in\R^{L\times d}$.
\For{$\ell=1{:}L_{\ell}$}
  \State For $m=1{:}M$ set $A^{(\ell,m)}\!=\!A(\Delta t^{(\ell,m)})$ and
  $k^{(\ell,m)}$ as above; define the mixture taps
  $k^{(\ell)}[\cdot]=\sum_{m=1}^{M} k^{(\ell,m)}[\cdot]\in\R^{d\times L_k}$.
  \State \emph{Depthwise causal convolution (per channel $c$):}
  \[
    U^{(\ell)}_{t,c}\;=\;\sum_{\tau=0}^{L_k-1} k^{(\ell)}_{c}[\tau]\;H^{(\ell-1)}_{t-\tau,c}.
  \]
  \State \emph{Squeeze--Excitation gate:}
  $s^{(\ell)}=\frac{1}{L}\sum_{t=1}^{L} H^{(\ell-1)}_{t,\cdot}$,
  $g^{(\ell)}=\sigma\!\big(W_2^{(\ell)}\phi(W_1^{(\ell)} s^{(\ell)})\big)$,
  $\widehat{H}^{(\ell)}_{t,\cdot}= U^{(\ell)}_{t,\cdot}\odot g^{(\ell)}$.
  \State \emph{Residual \& norm:} $Y^{(\ell)}=\mathrm{LN}\!\big(H^{(\ell-1)}+\widehat{H}^{(\ell)}\big)$.
  \State \emph{GLU mix:}
  $Z^{(\ell)}=W_{\downarrow}^{(\ell)}\!\big(\phi(W_{\uparrow,a}^{(\ell)}Y^{(\ell)})\odot
  \sigma(W_{\uparrow,g}^{(\ell)}Y^{(\ell)})\big)$,
  $H^{(\ell)}=\mathrm{LN}\!\big(Y^{(\ell)}+Z^{(\ell)}\big)$.
\EndFor
\State \emph{Head (last step):} $\hat{\mathbf{y}}= W_{\mathrm{head}}\;H^{(L_{\ell})}_{L,\cdot}\in\R^{O}$.

\State \textbf{Objective and update.}
\State Training loss on a batch $\mathcal B\subset\mathcal I_{\mathrm{tr}}$:
$\mathcal L_{\mathrm{tr}}(\theta)=\frac{1}{|\mathcal B|}\sum_{t\in\mathcal B}
\| f_\theta(\widetilde{\mathbf{X}}_t)-\widetilde{\mathbf{y}}_{t+1}\|_2^2$.
\State Compute gradient $g=\nabla_\theta \mathcal L_{\mathrm{tr}}(\theta)$; clip
$g\leftarrow g\cdot\min\!\big(1,c_{\max}/\|g\|\big)$; perform a weight-decayed step
$\theta\leftarrow \theta-\gamma_e\,(g+\lambda\,\theta)$.

\State \textbf{Early stopping and step-size scheduling.}
\State At epoch end, define
$\mathcal L_{\mathrm{va}}(e)=\frac{1}{|\mathcal I_{\mathrm{va}}|}\sum_{t\in\mathcal I_{\mathrm{va}}}
\| f_\theta(\widetilde{\mathbf{X}}_t)-\widetilde{\mathbf{y}}_{t+1}\|_2^2$.
If $\mathcal L_{\mathrm{va}}(e)$ improves the best by $>\!{\tt tol}$, set
$\theta^\star\!\leftarrow\!\theta$, reset counter; else increase counter.
When counter $\ge p$, stop. Optionally set
$\gamma_{e+1}\!\leftarrow\!\beta\,\gamma_e$ on plateaus ($\beta\in(0,1)$).

\State \textbf{Return.} $\theta^\star$, $(\bm\mu_x,\bm\sigma_x)$, $(\mu_y,\sigma_y)$.
\end{algorithmic}
\end{algorithm}

% ===================== Algorithm 2 (INFERENCE) — 1-column =====================
\begin{algorithm}[t!]
\footnotesize
\caption{MS$^{3}$M — Next-Step Inference}
\label{alg:ms3m_infer_math}
\begin{algorithmic}[1]
\Require Trained $\theta^\star$; scalers $(\bm\mu_x,\bm\sigma_x)$, $(\mu_y,\sigma_y)$;
new window $\mathbf{X}_{\mathrm{new}}\!\in\!\R^{L\times F}$
\Ensure $\widehat{\mathbf{y}}_{\mathrm{phys}}\in\R^{O}$

\State Standardize $\widetilde{\mathbf{X}}=(\mathbf{X}_{\mathrm{new}}-\bm\mu_x)\oslash\bm\sigma_x$.
\State Compute $\hat{\mathbf{y}}=f_{\theta^\star}(\widetilde{\mathbf{X}})$ as in lines 12–23 of Alg.~\ref{alg:ms3m_train_math}.
\State Inverse-standardize $\widehat{\mathbf{y}}_{\mathrm{phys}}=\mu_y+\sigma_y\,\hat{\mathbf{y}}$.
\State \Return $\widehat{\mathbf{y}}_{\mathrm{phys}}$.
\end{algorithmic}
\end{algorithm}

\subsection{Depthwise Convolution and Multi-Scale Mixture}\label{subsec:mixture}
\paragraph{Embedding}
We first embed standardized inputs to width $d$:
\begin{equation}\label{eq:embed}
H^{(0)}=\widetilde{\X}\,W_{\mathrm{in}}\in\R^{L\times d},
\qquad W_{\mathrm{in}}\in\R^{F\times d}
\quad\text{(Alg.~\ref{alg:ms3m_train_math}, line~10).}
\end{equation}

\paragraph{Mixture across $M$ time scales}
For layer $\ell\in\{1,\dots,L_\ell\}$ and component $m\in\{1,\dots,M\}$, we learn a positive step
$\Delta t^{(\ell,m)}=\phi(\tau^{(\ell,m)})$ (e.g., softplus of a raw parameter), form
$A^{(\ell,m)}=A(\Delta t^{(\ell,m)})$, and compute taps $k^{(\ell,m)}[\cdot]$ by \eqref{eq:taps}. We then \emph{sum} components:
\begin{equation}\label{eq:mixture-taps}
k^{(\ell)}[\cdot]\;=\;\sum_{m=1}^{M}k^{(\ell,m)}[\cdot]\in\R^{d\times L_k},
\qquad\text{(Alg.~\ref{alg:ms3m_train_math}, lines~11--12).}
\end{equation}
\emph{Rationale:} Distinct $\Delta t^{(\ell,m)}$ \emph{values} induce different decay/time constants, so \eqref{eq:mixture-taps} captures both fast and slow dynamics within the same receptive field.

\paragraph{Depthwise causal convolution}
With left zero-padding (strict causality), for channel $c\in\{1,\dots,d\}$ and time $t\in\{1,\dots,L\}$,
\begin{equation}\label{eq:dwconv}
U^{(\ell)}_{t,c}=\sum_{\tau=0}^{L_k-1}k^{(\ell)}_{c}[\tau]\;H^{(\ell-1)}_{t-\tau,c},
\quad\text{(Alg.~\ref{alg:ms3m_train_math}, line~13).}
\end{equation}
By Lemma~\ref{lem:decay}, $\|U^{(\ell)}_{\cdot,c}\|$ is stable and the truncation error is bounded by a geometric tail.

\subsection{Gating, Cross-Channel Mixing, and Normalization}\label{subsec:gating}
We modulate channels using SE gating and then apply a compact GLU mixer.

\paragraph{SE gate}
Let $s^{(\ell)}\!=\!\frac{1}{L}\sum_{t=1}^{L} H^{(\ell-1)}_{t,\cdot}\in\R^{d}$, choose a reduction $r\in\mathbb{N}$, and define
\begin{subequations}\label{eq:se}
\begin{align}
g^{(\ell)} &= \sigma\!\Big(W_2^{(\ell)}\,\phi\big(W_1^{(\ell)}\,s^{(\ell)}\big)\Big), \label{eq:se-def}\\
g^{(\ell)} &\in (0,1)^{d}, \label{eq:se-range}\\
W_1^{(\ell)} &\in \R^{d\times \lceil d/r\rceil},\quad
W_2^{(\ell)} \in \R^{\lceil d/r\rceil\times d}. \label{eq:se-shapes}
\end{align}
\end{subequations}
Apply $g^{(\ell)}$ channel-wise: $\widehat{H}^{(\ell)}_{t,\cdot}=U^{(\ell)}_{t,\cdot}\odot g^{(\ell)}$.

\paragraph{Residual and layer normalization}
Set
\begin{equation}\label{eq:res-ln1}
Y^{(\ell)}=\mathrm{LN}\!\big(H^{(\ell-1)}+\widehat{H}^{(\ell)}\big),
\end{equation}
where $\mathrm{LN}$ is per-time-step layer normalization:
$\mathrm{LN}(\mathbf{z})=\gamma\odot \frac{\mathbf{z}-\mu(\mathbf{z})}{\sigma(\mathbf{z})+\epsilon}+\beta$, with trainable $(\gamma,\beta)\in\R^{d}$.

\paragraph{GLU channel mixer}
We choose a hidden width $h=\alpha d$ with small $\alpha$ and define
\begin{equation}\label{eq:glu}
Z^{(\ell)}=W_{\downarrow}^{(\ell)}\!\Big(\phi(W_{\uparrow,a}^{(\ell)}Y^{(\ell)})
\odot \sigma(W_{\uparrow,g}^{(\ell)}Y^{(\ell)})\Big),
\end{equation}
with $W_{\uparrow,a}^{(\ell)},W_{\uparrow,g}^{(\ell)}\in\R^{d\times h}$ and
$W_{\downarrow}^{(\ell)}\in\R^{h\times d}$. Finalize the layer with
\begin{equation}\label{eq:res-ln2}
H^{(\ell)}=\mathrm{LN}\!\big(Y^{(\ell)}+Z^{(\ell)}\big),
\qquad\text{(Alg.~\ref{alg:ms3m_train_math}, lines~14--16).}
\end{equation}

\subsection{Prediction Head and Exact Mapping}\label{subsec:head-obj}
We read out only the representation at the most recent time index:
\begin{equation}\label{eq:head}
\begin{aligned}
\hat{\by} &= W_{\mathrm{head}}\, H^{(L_\ell)}_{L,\cdot} + b_{\mathrm{head}},\\[-2pt]
W_{\mathrm{head}} &\in \R^{d\times O},\quad
b_{\mathrm{head}} \in \R^{O},\quad
\hat{\by} \in \R^{O},
\end{aligned}
\end{equation}
as in Alg.~\ref{alg:ms3m_train_math}, line~17. The mapping
\[
(\widetilde{\X}\mapsto H^{(0)}\mapsto \{U^{(\ell)},\widehat{H}^{(\ell)},H^{(\ell)}\}_{\ell=1}^{L_\ell}\mapsto \hat{\by})
\]
is strictly causal \cite{gu2022s4}, linear in the convolutional part, and nonlinear only in channel-mixing/gating.

\subsection{Objective, Regularization, and Optimization}\label{subsec:obj}
For a batch $\mathcal{B}$, we minimize standardized MSE with weight decay (equivalently, an $\ell_2$ penalty):
\begin{equation}\label{eq:mse-obj}
\mathcal{L}_{\mathrm{tr}}(\theta)=\frac{1}{|\mathcal{B}|}\sum_{t\in\mathcal{B}}
\big\| f_{\theta}(\widetilde{\X}_{t}) - \widetilde{\by}_{t+1}\big\|_2^2
\;+\;\lambda\,\|\theta\|_2^2,
\end{equation}
(Alg.~\ref{alg:ms3m_train_math}, line~19). Parameters are updated with a clipped step
\begin{equation}\label{eq:clip-step}
g=\nabla_\theta \mathcal{L}_{\mathrm{tr}}(\theta),\quad
\tilde g=g\cdot\min\!\big(1,c_{\max}/\|g\|\big),\quad
\theta\leftarrow \theta-\gamma_e\,\tilde g,
\end{equation}
and early stopping on validation MSE (Alg.~\ref{alg:ms3m_train_math}, line~22).

\subsection{Inference and De-standardization}\label{subsec:infer}
Given a new window $\X_{\mathrm{new}}$, standardize it using train-only scalers, compute $\hat{\by}=f_{\theta^\star}(\widetilde{\X}_{\mathrm{new}})$ as above, and invert the target scaling
\begin{equation}\label{eq:destandardize}
\widehat{\by}_{\mathrm{phys}}=\mu_y+\sigma_y\,\hat{\by},
\qquad\text{(Alg.~\ref{alg:ms3m_infer_math}, lines~1--3).}
\end{equation}

\subsection{Identifiability \& Guarantees}\label{subsec:ident-guarantees}
We clarify what aspects of MS$^{3}$M are guaranteed by construction versus encouraged in practice, and which parameters are only identifiable up to benign symmetries.

\paragraph{Practical identifiability}
With standardized inputs/targets and a last-step head, trivial offsets are handled by $b_{\mathrm{head}}$. Depthwise SSM parameters $(B,C)$ and the embedding $W_{\mathrm{in}}$ admit classical scale/permutation symmetries (e.g., $C\!\leftarrow\!\alpha C$, $B\!\leftarrow\!\alpha^{-1}B$; channel permutations). In MS$^{3}$M these are \emph{mitigated}—not eliminated—by (i) separate parameter blocks with weight decay, (ii) per-time-step LayerNorm and SE gating, which fix effective channel scales, and (iii) the fixed wiring of SE/GLU and the head. Mixture time scales $\Delta t^{(\ell,m)}{>}0$ (via a positive map) are \emph{practically} identifiable up to component permutation (label-switching) and small rescalings in $(B,C)$.

\paragraph{Stability and boundedness}
Assume $A_{\ct}$ is Hurwitz. By bilinear (Tustin) discretization, $A(\Delta t)$ is Schur-stable for any $\Delta t{>}0$ (Prop.~\ref{prop:schur}). With finite kernel support $L_k$, each depthwise block is Finite Impulse Response (FIR) and hence Bounded-Input Bounded-Output (BIBO)-stable; with infinite support, Lemma~\ref{lem:decay} ensures geometric tail decay.

\paragraph{Causality}
All convolutions are one-sided with left zero-padding, so $f_\theta(\widetilde{\X}_t)$ is $\mathcal{F}_t$-measurable and uses no future inputs.

\paragraph{Truncation control}
Let $\|\cdot\|$ be a submultiplicative operator norm (e.g., induced 2-norm) and suppose $\|A(\Delta t)\|\le\alpha{<}1$. Then $\|k[\ell]\|\le \|C\|\,\|B\|\,\alpha^\ell$ for $\ell\ge1$, and the truncation error beyond $L_k$ is bounded by $\frac{\|C\|\,\|B\|}{1-\alpha}\alpha^{L_k}$ (Lemma~\ref{lem:decay}).

\paragraph{Expressivity}
Sums of stable exponentials (arising from the multi-scale mixture) form a rich dictionary that can approximate causal fading-memory filters on compact domains; SE/GLU provide cross-channel mixing without quadratic attention cost. See, e.g., classical Laguerre/Kautz expansions \cite{Oliveira2012PartII} for fading-memory approximation results.

\section{Data Collection and Preprocessing}
\label{sec:data}

\textbf{Testbed and logging:}
As shown in Figure \ref{fig:testbed}, we instrument an O\mbox{-}RAN stack with a near\mbox{-}RT RIC, a software \emph{Base Station} (BS)/UE using \emph{srsRAN} on \emph{Universal Software Radio Peripheral} (USRP) \emph{Software\mbox{-}Defined Radios} (SDRs), a video server (MediaMTX + FFmpeg), and a controllable interferer. Unless stated otherwise, the BS operates at 2680\,MHz downlink, 25 \emph{Physical Resource Blocks} (PRBs), $2{\times}2$ \emph{Multiple\mbox{-}Input Multiple\mbox{-}Output} (MIMO), \emph{Frequency\mbox{-}Division Duplexing} (FDD). The interferer is implemented in C\texttt{++}/\emph{USRP Hardware Driver} (UHD) and injects random\mbox{-}length \emph{Orthogonal Frequency\mbox{-}Division Multiplexing} (OFDM) bursts with randomized gain, sleeping \([1,5]\)~ms between bursts to create intermittent co\mbox{-}channel interference. Each run lasts 120\,s while the UE streams video. We log (i) \emph{Physical} (PHY)/\emph{Medium Access Control} (MAC) KPIs from BS/UE (exported every 20\,ms), (ii) FFmpeg streaming statistics, and (iii) optional packet captures. Two deployments are used: a cloudified next\mbox{-}generation (xG) testbed (USRP X310; OpenStack \emph{Virtual mMchines} (VMs)) and a lab setup (USRP B210; standalone hosts). Table \ref{tab:dataset_stats_units} summarizes the dataset statistics for the considered KPIs. The dataset\footnote{A version of our dataset is available at: \url{https://ieee-dataport.org/documents/video-streaming-network-kpis-o-ran-testing}} is publicly available, which ensures reproducibility and facilitates further research.

\begin{figure}[H]
\centering
\includegraphics[width=0.98\linewidth]{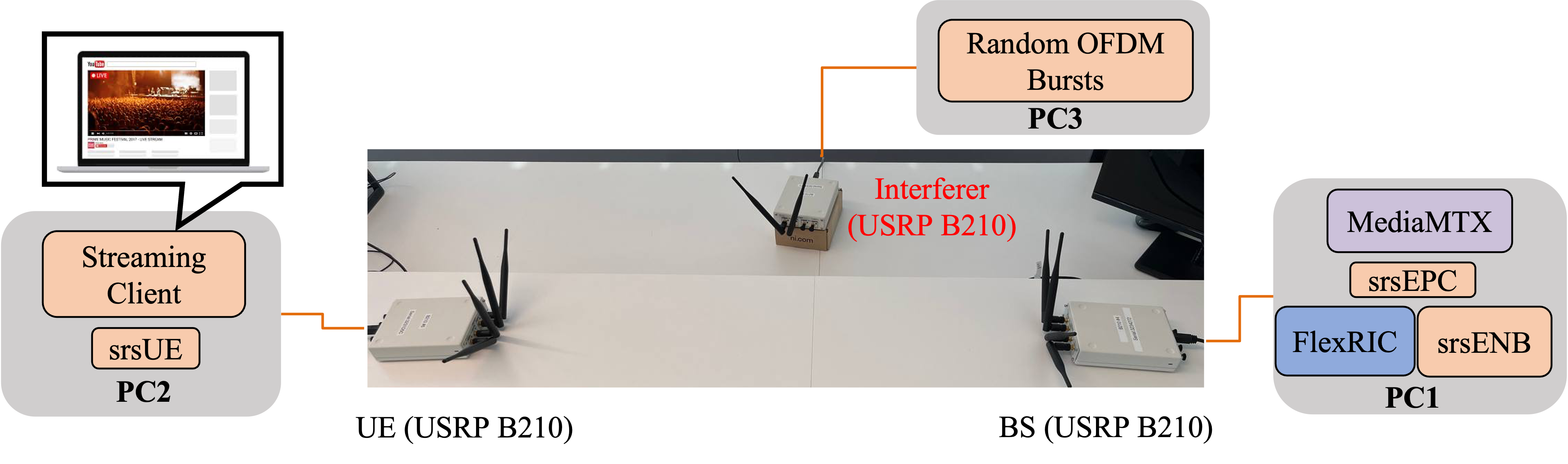}
\caption{Virginia Tech Innovation Campus O-RAN testbed setup \cite{dai2025orankpi}.}
\label{fig:testbed}
\end{figure}
\textbf{Temporal aggregation to a uniform grid:}
Heterogeneous logs are aligned by windowed averaging onto a common timeline. For KPI $k$ with raw samples $\{(t_i,x_i^{(k)})\}$, a window of length $\Delta$ and stride $\tau$ produces
\begin{equation}
\bar{x}^{(k)}(t_m)=\frac{1}{N_m}\!\!\sum_{i:\,t_i\in[t_m,\,t_m+\Delta)}\!\!x_i^{(k)},\qquad t_{m+1}=t_m+\tau,
\end{equation}
yielding uniformly spaced series $\{\bar{x}^{(k)}(t_m)\}_m$. We left\mbox{-}join all KPIs at $t_m$ to form a matrix $D_{\text{kpi}}$ (one row per time step).

\textbf{Semantics\mbox{-}aware missingness:}
Because no packet may arrive in some intervals, ``UE Packet Delay'' can be absent while other KPIs remain valid. We therefore \emph{retain} rows missing \emph{only} this field and impute a reserved sentinel $-1$; any row missing \emph{other} KPIs is \emph{dropped}, as it indicates a broader measurement gap.

\textbf{Outlier control (IQR pruning):}
Per KPI, define $Q_1=\mathrm{quantile}_{0.10}$, $Q_3=\mathrm{quantile}_{0.90}$, $\mathrm{IQR}=Q_3{-}Q_1$. Samples outside
\begin{equation}
[\,Q_1-1.5\,\mathrm{IQR},\; Q_3+1.5\,\mathrm{IQR}\,]
\end{equation}
are removed prior to forming sequences. This decile\mbox{-}based rule dampens heavy tails without erasing typical dynamics.

\textbf{Sequential samples for one\mbox{-}step forecasting:}
From the cleaned table with $K$ KPIs, we extract contiguous sequences of length $N_{\text{seq}}$ (we use $N_{\text{seq}}{=}28$) only if all within\mbox{-}window gaps equal $\tau$. For each valid index $m$,

\begin{equation}
\begin{aligned}
X_m &= \big[\bar{\bm{x}}(t_{m-N_{\mathrm{seq}}+1}),\,\ldots,\,\bar{\bm{x}}(t_m)\big],\\[-2pt]
y_m &= \bar{\bm{x}}(t_{m+1}).
\end{aligned}
\label{eq:seqpair}
\end{equation}
\noindent\textit{Types: } $X_m \in \mathbb{R}^{N_{\mathrm{seq}}\times K}$ and $y_m \in \mathbb{R}^{K}$.

Yielding paired datasets $D_x=\{X_m\}_{m=1}^{M}$ and $D_y=\{y_m\}_{m=1}^{M}$. For the RSRP task, we select the appropriate component of $y_m$ as the target; all covariates used by forecasters are shifted by $+1$ step downstream to guarantee leakage safety.

\textbf{KPI selection and splits:}
KPI choices follow the O\mbox{-}RAN End\mbox{-}to\mbox{-}End Test Specification~\cite{oran-e2e}. We use contiguous tail splits for train/validation/test and fit any scalers or feature selection on \emph{train only} to prevent information leakage.

\begin{table}[t!]
\centering
\scriptsize
\setlength{\tabcolsep}{3pt}
\renewcommand{\arraystretch}{1.12}

\caption{Statistical overview of measured performance indicators, with units and interpretations \cite{Rezazadeh2025ReservoirAugmented}.}
\label{tab:dataset_stats_units}

\begin{tabular*}{\columnwidth}{@{\extracolsep{\fill}} l c p{0.35\columnwidth} r r r r}
\toprule
& \multicolumn{2}{c}{\textbf{Definition}}
& \multicolumn{4}{c}{\textbf{Summary Statistics}} \\
\cmidrule(lr){2-3}\cmidrule(lr){4-7}
\textbf{Feature} & \textbf{Unit} & \textbf{Meaning}
& \textbf{Min} & \textbf{Max} & \textbf{Mean} & \textbf{Std} \\
\midrule
MCS    & index & Modulation and coding level chosen for transmission & 0.00 & 26.67 & 9.06 & 4.94 \\
CQI    & index & UE feedback on perceived downlink channel quality   & 0.00 & 13.00 & 8.51 & 0.92 \\
RI     & rank  & Number of spatial layers scheduled (MIMO rank)      & 0.00 & 2.00  & 1.36 & 0.38 \\
PMI    & index & UE’s preferred precoding matrix indicator           & 0.00 & 3.00  & 0.93 & 0.84 \\
Buffer & bytes & Data volume queued in the UE uplink buffer          & 0.00 & 3437.67 & 25.81 & 88.21 \\
RSRQ   & dB    & Signal quality based on reference signal and received power & -14.00 & -6.40 & -10.54 & 2.47 \\
RSRP   & dBm   & Average received reference-signal power             & -104.67 & -75.00 & -87.59 & 3.70 \\
RSSI   & dBm   & Total wideband received signal strength             & -70.00 & -60.00 & -65.37 & 2.62 \\
SINR   & dB    & Ratio of useful signal power to interference plus noise & 1.10 & 24.33 & 18.31 & 1.92 \\
PRBs   & RBs   & Number of physical resource blocks allocated        & 2.00 & 25.00 & 22.30 & 4.36 \\
SE     & bps/Hz & Throughput efficiency per unit of spectrum         & 0.00 & 3.74  & 0.58 & 0.39 \\
BLER   & \%    & Fraction of incorrectly received transport blocks   & 0.00 & 78.00 & 2.63 & 6.76 \\
Delay  & ms    & End-to-end packet transmission latency              & 1.00 & 9861.82 & 63.15 & 219.76 \\
\bottomrule
\end{tabular*}

\vspace{2pt}
\footnotesize\emph{Notes.} Units follow standard 4G/5G KPI conventions. Statistics computed over the full evaluation set.
\end{table}

% =========================================================
\section{Transformer Baselines}\label{sec:xf-bench}
To position the proposed \emph{MS\textsuperscript{3}M} against widely used neural forecasters, we establish a transparent, leakage\mbox{-}safe benchmark comprising seven strong Transformer\mbox{-}era models implemented under a single pipeline, as shown in Tables \ref{tab:master_perf_complex} and \ref{tab:master_config_hps}: \emph{FEDformer}~\cite{zhou2022fedformer}, \emph{Informer}~\cite{zhou2021informer}, \emph{TFT} (Temporal Fusion Transformer)~\cite{lim2021temporal}, \emph{ETSformer}~\cite{woo2022etsformer}, \emph{Crossformer}~\cite{zhang2023crossformer}, \emph{PatchTST}~\cite{nie2023patchtst}, and \emph{iTransformer}~\cite{liu2024itransformer}. This section details the protocol, model settings, metrics, complexity reporting, and fairness considerations so that results can be interpreted unambiguously and reproduced by non\mbox{-}specialists.

\subsection{Rationale and Scope}
The chosen baselines cover complementary inductive biases for multivariate time\mbox{-}series forecasting: frequency\mbox{-}domain and decomposition (FEDformer), efficient long\mbox{-}sequence attention (Informer), feature\mbox{-}aware attention with gating and variable selection (TFT), exponential\mbox{-}smoothing\mbox{-}inspired decomposition and attention (ETSformer), explicit cross\mbox{-}dimension \& cross\mbox{-}time modeling (Crossformer), channel\mbox{-}independent patching (PatchTST), and inverted dimension attention over variate tokens (iTransformer).
Using this suite spans (i) seasonal/trend decomposition, (ii) long\mbox{-}context efficiency, (iii) multivariate cross\mbox{-}dimension structure, and (iv) channel\mbox{-}independent tokenization.

\subsection{Data Handling and Leakage Prevention}\label{sec:leakage}
We consider one\mbox{-}step\mbox{-}ahead forecasting with a fixed \emph{lookback} window of $W{=}32$ past steps and horizon $H{=}1$. To guarantee leakage safety:
\begin{itemize}[leftmargin=*,nosep]
  \item \textbf{Past\mbox{-}only covariates.} All covariates (KPIs) used to predict $y_t$ are shifted by exactly one step so they are measurable at time $t{-}1$. No contemporaneous or future information enters the predictors for $y_t$.
  \item \textbf{Contiguous tail splits.} After the shift, the remaining series is partitioned into contiguous \emph{train}, \emph{validation}, and \emph{test} tails (validation fraction $15\%$, test fraction $15\%$) to respect temporal order.
  \item \textbf{Train\mbox{-}only standardization.} Scaling Transformers are fit on the training split only and then applied to validation and test (for targets and covariates separately).
\end{itemize}
These choices mirror best practice in time\mbox{-}series evaluation and avoid optimistic bias from look\mbox{-}ahead leakage. For the channel\mbox{-}independent variant (PatchTST), we additionally report results under \emph{RSRP\mbox{-}only} input (Channel-isolated (CI)) to align with its recommended usage.

\begin{table*}[t]
\centering
\setlength{\tabcolsep}{3.2pt}
\caption{Comprehensive Comparison of Forecasters: \emph{Performance} and \emph{Complexity}.}
\label{tab:master_perf_complex}
\resizebox{\textwidth}{!}{%
\begin{tabular}{l
                ccccc
                ccccc}
\toprule
& \multicolumn{5}{c}{\textbf{Performance (Test Tail)}} 
& \multicolumn{5}{c}{\textbf{Complexity}} \\
\cmidrule(lr){2-6}\cmidrule(lr){7-11}
\textbf{Method} 
& \textbf{RMSE (dB)} & \textbf{MAE (dB)} & \textbf{MSE (dB$^2$)} & \textbf{Skill\,(R)} & \textbf{Skill\,(M)} 
& \textbf{\#Params} & \textbf{Infer (s)} & \textbf{MS$^{3}$M Speedup (×)} & \textbf{Dominant (our setting)} & \textbf{W/H} \\
\midrule

\multicolumn{11}{l}{\textbf{Proposed Model}} \\
\addlinespace[2pt]
MS\textsuperscript{3}M
& 0.292 & 0.170 & 0.090 & +92.3\% & +94.0\%
& 698{,}449 & 0.057 & 1.00 & $\mathcal{O}(L\, d\, d_s)$
& 32/1 \\
\addlinespace[3pt]
\midrule

\multicolumn{11}{l}{\textbf{Baselines}} \\
\addlinespace[2pt]
\textbf{FEDformer}~\cite{zhou2022fedformer}
& 0.599 & 0.394 & 0.359 & +84.12\% & +86.19\%
& 755{,}906 &  0.415 & 7.28 & $\mathcal{O}(M\, d^{2})$
& 32/1 \\
\addlinespace[2pt]
\textbf{Informer}~\cite{zhou2021informer}
& 0.368 & 0.194 & 0.135 & +90.25\% & +93.21\%
& 1{,}256{,}449 &  0.298 & 5.23 & $\mathcal{O}(L\, d^{2})$
& 32/1 \\
\addlinespace[2pt]
\textbf{TFT}~\cite{lim2021temporal}
& 0.422 & 0.241 & 0.178 & +88.82\% & +91.57\%
& 2{,}510{,}702 & 0.229 & 4.02 & $\mathcal{O}(L\, d_h^{2})$ \;($d_h{=}256$)
& 32/1 \\
\addlinespace[2pt]
\textbf{ETSformer}~\cite{woo2022etsformer}
& 0.333 & 0.231 & 0.111 & +91.16\% & +91.92\%
& 10{,}376 &  0.192 & 3.37 & {$\mathcal{O}(L\, d + K\, d \log L)$}
& 32/1 \\
\addlinespace[2pt]
\textbf{Crossformer}~\cite{zhang2023crossformer}
& 0.275 & 0.154 & 0.076 & +92.70\% & +94.60\%
& 1{,}591{,}321 &  0.586 & 10.28 & $\mathcal{O}(L\, d^{2})$
& 32/1 \\
\addlinespace[2pt]
\textbf{PatchTST}~\cite{nie2023patchtst}
& 3.197 & 2.500 & 10.218 & +15.27\% & +28.20\%
& 662{,}403 &  0.233 & 4.09 & $\mathcal{O}(N_p\, d^{2})$
& 32/1 \\
\addlinespace[2pt]
\textbf{iTransformer}~\cite{liu2024itransformer}
& 3.396 & 2.690 & 11.533 & +9.97\% & +5.71\%
& 814{,}273 &  0.217 & 3.81 & $\mathcal{O}(F\, d^{2})$
& 32/1 \\
\addlinespace[3pt]
\midrule

\multicolumn{11}{l}{\footnotesize \textbf{Notes:} “MS$^{3}$M Speedup (×)” = $(\text{Latency of method}) / (\text{Latency of MS$^{3}$M})$; higher is better for MS$^{3}$M. Skills are relative to persistence $y_{t-1}$.}
\\
\bottomrule
\end{tabular}%
}
\end{table*}
% =======================================================================================

% ======================= Results Table — Configurations, Training HPs, Safety =======================
\begin{table*}[t]
\centering
\setlength{\tabcolsep}{3.2pt}
\caption{Under the Hood of Forecasters: \emph{Configuration}, \emph{Training}, and \emph{Uncertainty}}
\label{tab:master_config_hps}
\resizebox{\textwidth}{!}{%
\begin{tabular}{l
                cc
                cc
                cc}
\toprule
& \multicolumn{2}{c}{\textbf{Model Configuration}} 
& \multicolumn{2}{c}{\textbf{Training HPs}} 
& \multicolumn{2}{c}{\textbf{Safety \& Uncertainty}} \\
\cmidrule(lr){2-3}\cmidrule(lr){4-5}\cmidrule(lr){6-7}
\textbf{Method} 
& \textbf{Exog/State} & \textbf{Arch / Key Dims} 
& \textbf{LR / Pat.} & \textbf{MaxEp / Batch} 
& \textbf{Leakage-Safe} & \textbf{Uncertainty} \\
\midrule

\multicolumn{7}{l}{\textbf{Proposed Model}} \\
\addlinespace[2pt]
MS\textsuperscript{3}M
& $F{=}13,\ d_{\mathrm{state}}{=}64$ 
& S6Mix$\times$4,\ $d_{\mathrm{model}}{=}128$;\ $N{=}64$;\ $m{=}4$;\ drop 0.1
& $2\!\times\!10^{-3}$/20 & 60 / 256 
& Yes & Point \\
\addlinespace[3pt]
\midrule

\multicolumn{7}{l}{\textbf{Baselines}} \\
\addlinespace[2pt]
\textbf{FEDformer}
& $F{=}13$
& Enc$\times$3, Dec$\times$2;\ $d_{\mathrm{model}}{=}128$;\ modes 32;\ MA $\{3,5,7,11,25\}$;\ drop 0.10
& $2\!\times\!10^{-3}$/20 & 60 / 256
& Yes & Point \\
\addlinespace[2pt]
\textbf{Informer}
& $F{=}13$
& Enc$\times$3, Dec$\times$2;\ $d_{\mathrm{model}}{=}128$;\ $n_{\mathrm{head}}{=}8$;\ ProbSparse $c{=}5$;\ distill: Yes;\ PE: learned;\ label\_len 16;\ pred\_len 1;\ drop 0.10;\ causal mask
& $2\!\times\!10^{-3}$/20 & 60 / 256
& Yes & Point \\
\addlinespace[2pt]
\textbf{TFT}
& $F{=}13,\ \text{static}=0$
& VSN;\ LSTM enc/dec;\ GRNs;\ Interp-MHAttn;\ $d_{\mathrm{model}}{=}128$;\ $d_{\mathrm{hidden}}{=}256$;\ $n_{\mathrm{head}}{=}8$;\ drop 0.10;\ attn-drop 0.10;\ quantiles $\{0.1,0.5,0.9\}$
& $2\!\times\!10^{-3}$/20 & 60 / 256
& Yes & Quantiles \\
\addlinespace[2pt]
\textbf{ETSformer}
& $F{=}13$
& ETS layers$\times$3;\ $d_{\mathrm{model}}{=}128$;\ Top-$K{=}8$ (freq);\ max\_lag 16;\ ES baseline: Yes;\ drop 0.10/attn 0.10
& $2\!\times\!10^{-3}$/20 & 60 / 256
& Yes & Point \\
\addlinespace[2pt]
\textbf{Crossformer}
& $F{=}13$
& TSA (routers $c{=}8$);\ DSW: patch 16/stride 8;\ HED: S1$\times$2, S2$\times$1, merge$\times$2;\ $d_{\mathrm{model}}{=}128$;\ $n_{\mathrm{head}}{=}8$;\ drop 0.10/attn 0.10
& $2\!\times\!10^{-3}$/20 & 60 / 256
& Yes & Point \\
\addlinespace[2pt]
\textbf{PatchTST}
& RSRP-only (CI)
& $d_{\mathrm{model}}{=}128$;\ $n_{\mathrm{head}}{=}16$;\ layers 3;\ $d_{\mathrm{ff}}{=}256$;\ drop 0.20;\ patch 16;\ stride 8;\ RevIN;\ BN-enc
& $2\!\times\!10^{-3}$/20 & 60 / 256
& Yes & Point \\
\addlinespace[2pt]
\textbf{iTransformer}
& $F{=}13$
& $d_{\mathrm{model}}{=}128$;\ $n_{\mathrm{head}}{=}8$;\ layers 4;\ token-embed: No
& $2\!\times\!10^{-3}$/20 & 60 / 256
& Yes & Point \\
\addlinespace[3pt]
\midrule

\multicolumn{7}{l}{\footnotesize \textbf{Notes:} All models use the same leakage-safe pipeline. “CI” = channel-isolated input.}
\\
\bottomrule
\end{tabular}%
}
\end{table*}
% =======================================================================================

\subsection{Baseline Models}\label{sec:models}
All models share the same lookback/horizon and data pipeline for comparability ($W{=}32$, $H{=}1$, $F{=}13$ exogenous KPIs unless CI is stated). Architectural settings follow the common, compute\mbox{-}matched configuration in Table~\ref{tab:master_config_hps}:
\begin{itemize}[leftmargin=*,nosep]
  \item \textbf{FEDformer}: seasonal\mbox{-}trend decomposition with Fourier/Wavelet enhanced blocks; encoder$\times$3, decoder$\times$2, $d_{\mathrm{model}}{=}128$.
  \item \textbf{Informer}: ProbSparse self\mbox{-}attention with distillation; encoder$\times$3, decoder$\times$2, $d_{\mathrm{model}}{=}128$, $n_{\mathrm{head}}{=}8$, causal mask.
  \item \textbf{TFT}: variable selection, gating, LSTM encoder/decoder, interpretable multi\mbox{-}head attention; $d_{\mathrm{model}}{=}128$, $d_{\mathrm{hidden}}{=}256$, $n_{\mathrm{head}}{=}8$; quantile head $\{0.1,0.5,0.9\}$.
  \item \textbf{ETSformer}: exponential\mbox{-}smoothing attention and frequency attention; $d_{\mathrm{model}}{=}128$, Top-$K{=}8$, max\_lag 16.
  \item \textbf{Crossformer}: DSW patching + two\mbox{-}stage attention with hierarchical encoder\mbox{-}decoder; $d_{\mathrm{model}}{=}128$, $n_{\mathrm{head}}{=}8$, patch 16/stride 8.
  \item \textbf{PatchTST} (CI): channel\mbox{-}independent patches (RSRP\mbox{-}only), $d_{\mathrm{model}}{=}128$, $n_{\mathrm{head}}{=}16$, layers 3, $d_{\mathrm{ff}}{=}256$, patch 16/stride 8, RevIN.
  \item \textbf{iTransformer}: inverted dimensions with variate tokens (\emph{uses all $F{=}13$ KPIs as inputs}), $d_{\mathrm{model}}{=}128$, $n_{\mathrm{head}}{=}8$, and  4 layers.
\end{itemize}
We intentionally avoid per\mbox{-}model hyperparameter sweeps to keep budgets aligned; this is conservative for neural baselines. For clarity, we define each column shown in the Tables~\ref{tab:master_perf_complex} and \ref{tab:master_config_hps}:
\begin{itemize}[leftmargin=*,nosep]
  \item \emph{Performance (Test Tail)}: RMSE/MAE/MSE on the contiguous test tail; \emph{Skill\,(R/M)} are relative to persistence ($y_{t-1}$), per \S\ref{sec:metrics}.
  \item \emph{Complexity}: \#Params = number of trainable parameters; \emph{Infer (s)} = wall-clock seconds for a single forward pass over the test tail (no I/O), after a short warm-up, on the shared device; \emph{Dominant (our setting)} = leading-order time complexity of a forward pass under our shared window/hyperparameter setting (symbols: $L$ lookback length, $H$ horizon, $d$ model width, $d_h$ hidden width, $d_{\mathrm{state}}$ state size, $F$ exogenous feature count, $M$ spectral modes, $K$ Top-$K$ frequencies, $N_p$ number of patches).
  \item \emph{Model Configuration}: \emph{W/H} = lookback window / forecast horizon; \emph{Exog/State} = number of exogenous inputs $F$ (and, where applicable, model state size, e.g., $d_{\mathrm{state}}$ for MS$^{3}$M); \emph{Arch / Key Dims} lists depth (e.g., Enc$\times L$, Dec$\times M$), $d_{\mathrm{model}}$, $n_{\mathrm{head}}$, patch/stride, etc.
  \item \emph{Training HPs}: \emph{LR / Pat.} = Adam learning rate and early-stopping \emph{patience} (in epochs) based on validation loss; \emph{MaxEp / Batch} = maximum epochs allowed before early stop and mini-batch size. These govern \emph{training only}; inference timings above are independent of them.
  \item \emph{Safety \& Uncertainty}: \emph{Leakage-Safe} = ``Yes'' if the three safeguards in \S\ref{sec:leakage} (past-only covariates, contiguous tail splits, train-only scaling) are enforced; \emph{Uncertainty} indicates the prediction type: \emph{Point} (mean/point forecast) or \emph{Quantiles} (e.g., TFT reports $\{0.1,0.5,0.9\}$).
\end{itemize}

\subsection{Training Protocol}\label{sec:training}
All models are optimized with Adam (learning rate $2{\times}10^{-3}$), batch size $256$, maximum $60$ epochs, and early stopping on validation loss (patience $=20$). A fixed random seed ensures determinism where supported (some GPU atomics may still introduce tiny run-to-run differences). Training uses a single GPU if available; otherwise, it uses the CPU. Parameter counts (\#Params) include only trainable tensors.

\subsection{Metrics and Skill Relative to Persistence}\label{sec:metrics}
We report accuracy on the contiguous test tail using RMSE and MAE. For a set of test timestamps $\mathcal{T}$ (size $N$),
\begin{equation}
\mathrm{RMSE}(\hat y,y)=\sqrt{\frac{1}{N}\sum_{t\in\mathcal{T}}(\hat y_t-y_t)^2},
\end{equation}
\begin{equation}
\mathrm{MAE}(\hat y,y)=\frac{1}{N}\sum_{t\in\mathcal{T}}\lvert \hat y_t-y_t\rvert .
\end{equation}
To contextualize absolute errors, we also report \emph{skill} against a strict, leakage\mbox{-}safe \emph{persistence} baseline that predicts $y_{t-1}$ at time $t$ using the same test timestamps and alignment:
\begin{equation}
\label{eq:skill}
\begin{aligned}
\text{Skill}_{\mathrm{RMSE}} &= 1-\frac{\mathrm{RMSE}(\hat y,y)}{\mathrm{RMSE}(y_{t-1},y_t)},\\
\text{Skill}_{\mathrm{MAE}}  &= 1-\frac{\mathrm{MAE}(\hat y,y)}{\mathrm{MAE}(y_{t-1},y_t)}.
\end{aligned}
\end{equation}

\noindent\textit{Interpretation:} A skill of $0$ means parity with persistence; values $>0$ indicate improvement (larger is better); negative values indicate worse than persistence.

\subsection{Computational Footprint}\label{sec:complexity}
We report two complementary indicators of computational cost: (i) the number of trainable parameters (\#Params) and (ii) observed \emph{test-time inference} wall-clock on the common setup of \S\ref{sec:training}. Because absolute times are hardware-dependent, we present the measured inference latency alongside \#Params to convey both asymptotic and practical cost. Unless stated otherwise, inference time is for a single forward pass over the contiguous test tail (data already in memory), after a short warm-up.

\subsubsection{Costs and Inference Complexity}
Costs are reported as functions of $(L,F,d,H,E,D,P,S,M,K,d_s)$, where
, where
$L$ = lookback, $F$ = \#inputs, $d$ = model width, $H$ = \#heads,
$E,D$ = encoder/decoder layers, $P,S$ = patch length/stride,
$M$ = retained modes, $K$ = Top\mbox{-}K, and $d_s$ = state size.
Peak activation memory follows the attention term:
$\Theta(L^2)$ for full self\mbox{-}attention,
$\Theta(N_p^2)$ with patches (where $N_p \!\approx\! \lceil (L-P)/S\rceil{+}1$),
and $\Theta(L\,d)$ for state\mbox{-}space mixers (SSMs).

\begin{itemize}[leftmargin=*,nosep]
  \item \textbf{MS$^{3}$M}: Per layer, linear-time sequence mixing scales as $\mathcal{O}(L\,d\,d_s)$ compute and $\mathcal{O}(L\,d)$ memory. With $L{=}32$, $d{=}128$, $d_s{=}64$, and 4 layers, the forward pass is $\mathcal{O}(4\,L\,d\,d_s)$.

  \item \textbf{FEDformer}: With frequency-domain blocks retaining $M$ modes, per layer cost is $\mathcal{O}(L\,d\log L + M\,d^2)$; total over $E{+}D$ layers is $\mathcal{O}((E{+}D)(L\,d\log L + M\,d^2))$.
  \item \textbf{Informer}: ProbSparse attention reduces full $L^2$ attention to $\mathcal{O}(c\,L\,d)$ with $c\!\ll\!L$ influential queries; including projections the per-layer cost is $\mathcal{O}(c\,L\,d + L\,d^2)$, and total is $\mathcal{O}((E{+}D)(c\,L\,d + L\,d^2))$.
  \item \textbf{TFT}: LSTM encoder/decoder contributes $\mathcal{O}(L\,d_h^2)$ per stack (with $d_h$ the LSTM hidden size), and the interpretable multi-head attention adds $\mathcal{O}(L^2 d)$ per attention block. Net cost is $\mathcal{O}(L\,d_h^2 + L^2 d)$ per layer group.
  \item \textbf{ETSformer}: Exponential-smoothing attention and frequency attention yield per-layer cost $\mathcal{O}(L\,d + K\,d\log L)$; across $E$ layers the forward pass is $\mathcal{O}(E(L\,d + K\,d\log L))$.
  \item \textbf{Crossformer}: With DSW patching (patch $P$, stride $S$), the number of patches $N_p \approx \lceil (L-P)/S\rceil{+}1$. Two-stage attention over patches scales per layer as $\mathcal{O}(N_p^2 d) + \mathcal{O}(L d)$ (within-patch ops). Total over $E{+}D$ layers is $\mathcal{O}((E{+}D)(N_p^2 d + L d))$.
  \item \textbf{PatchTST}: Encoder-only with channel-isolated patches. Per layer, MHSA over $N_p$ patches costs $\mathcal{O}(N_p^2 d)$ (plus projections $\mathcal{O}(N_p d^2)$). With $F{=}1$ (CI) and $E$ layers, total is $\mathcal{O}(E(N_p^2 d + N_p d^2))$.
  \item \textbf{iTransformer}: Variables-as-tokens, sequence length is $F$ (not $L$). Per layer attention is $\mathcal{O}(F^2 d)$; the temporal mixing/projection across $L$ adds $\mathcal{O}(L\,F\,d)$ (linear ops). With $E$ layers, total is $\mathcal{O}(E(F^2 d + L F d))$.
\end{itemize}

\subsubsection{Order-of-growth summary (lower $\to$ higher)}The chain ranks models by their \textbf{dominant per-layer inference cost} under our chosen hyperparameters, comparing only leading terms (Big-$O$), i.e., ignoring constants and lower-order terms. The symbol $\lesssim$ means “no larger up to constants” (same order or smaller), while $\approx$ groups models in the same cost tier. This ranking reflects \emph{asymptotic} compute; practical wall-clock can differ due to hardware, kernels, memory bandwidth, and implementation details. With $L{=}32,\,F{=}13,\,d{=}128,\,d_h{=}256,\,P{=}16,\,S{=}8 \Rightarrow N_p{=}3,\,M{=}32,\,K{=}8,\,d_s{=}64$, we have $L d^2 = M d^2$, placing Informer, FEDformer, and Crossformer in the same tier:\\

\noindent\fbox{%
  \parbox{\columnwidth}{%
    \centering
    \(\displaystyle
    \text{ETSformer} \lesssim \text{PatchTST} \lesssim \text{iTransformer} \lesssim \textbf{MS}^{3}\textbf{M} \lesssim\; \text{Informer} \approx \text{FEDformer} \approx \text{Crossformer} \;\lesssim\; \text{TFT}
    \)
  }%
}
\subsection{Fairness and Threats to Validity}\label{sec:fairness}
All baselines share identical lookback/horizon, covariate sets, splits, scaling, optimizer, batching, and early stopping. We do \emph{not} run model\mbox{-}specific sweeps; therefore, the comparison is \emph{budget\mbox{-}matched} but may be conservative for certain neural models that benefit from tuning. Reported wall times are indicative rather than absolute. The leakage\mbox{-}safe design in \S\ref{sec:leakage} avoids information bleed; consequently, discrepancies relative to published leaderboard numbers (often using longer horizons, future covariates, or different tokenization) are expected. CI (RSRP\mbox{-}only) results are highlighted only for PatchTST, which aligns with its intended channel\mbox{-}independent usage.

\subsection{Reproducibility Checklist}\label{sec:repro}
We release: (i) complete preprocessing specifications (shifted covariates; contiguous tail splits; train\mbox{-}only scaling), (ii) fixed hyperparameters per model (see \S\ref{sec:models} and Table~\ref{tab:master_config_hps}), (iii) the random seed and device policy (\S\ref{sec:training}), and (iv) exported file summaries matching Table~\ref{tab:master_perf_complex}. The baselines uses a single, self\mbox{-}contained pipeline so that another practitioner can regenerate numbers without manual intervention.

\section{Performance and Numerical Results}\label{sec:results}

Table~\ref{tab:master_perf_complex} summarizes accuracy, efficiency, and model size across methods. Overall, \textbf{MS\textsuperscript{3}M} offers the best accuracy–efficiency balance, achieving \textbf{RMSE 0.292}, \textbf{MAE 0.170}, and \textbf{MSE 0.090} with skill gains of \textbf{+92.3\%} (R) and \textbf{+94.0\%} (M). It is also the \emph{fastest} at inference (\textbf{0.057\,s}) with a compact footprint (698{,}449 parameters).
% ====== Table: Point metrics with 95% bootstrap CIs ======
\begin{table}[t]
  \centering
  \caption{Test-set performance in original dBm units with 95\% bootstrap confidence intervals.}
  \label{tab:s6_metrics}
  \begin{tabular}{lccc}
    \toprule
    Metric & Point & 95\% CI (low) & 95\% CI (high) \\
    \midrule
    RMSE (dBm) & 0.290 & 0.279 & 0.300 \\
    MAE  (dBm) & 0.169 & 0.164 & 0.174 \\
    $R^2$      & 0.9931 & 0.9925 & 0.9937 \\
    \bottomrule
  \end{tabular}
\end{table}
\paragraph{Comparison to strong baselines}
\emph{Crossformer} attains slightly lower raw errors (0.275/0.154/0.076 for RMSE/MAE/MSE) but is more than an order of magnitude slower (0.586\,s) and over twice as large (1{,}591{,}321 parameters), yielding a less attractive Pareto trade-off for real-time or resource-constrained use. 
Against \emph{FEDformer}, MS\textsuperscript{3}M reduces error across all metrics (0.292/0.170/0.090 vs.\ 0.599/0.394/0.359), runs faster at inference (0.057\,s vs.\ 0.415\,s), and uses fewer parameters (698k vs.\ 756k). 
Relative to \emph{Informer}, MS\textsuperscript{3}M achieves lower errors (0.292/0.170/0.090 vs.\ 0.368/0.194/0.135), higher skill (+92.3\%/+94.0\% vs.\ +90.25\%/+93.21\%), and markedly lower latency (0.057\,s vs.\ 0.298\,s), despite Informer’s ProbSparse attention, learned positional encodings, and distillation.
Compared with \emph{TFT}, which uses rich gating/attention with quantile outputs, MS\textsuperscript{3}M attains lower errors (0.292/0.170/0.090 vs.\ 0.422/0.241/0.178), \mbox{$\sim$4$\times$} faster inference (0.057\,s vs.\ 0.229\,s), and far fewer parameters (0.70M vs.\ 2.51M). 
\emph{ETSformer} is impressively lightweight (10{,}376 parameters) with competitive skills (+91.16\%/+91.92\%), yet MS\textsuperscript{3}M delivers better tail errors (0.292/0.170/0.090 vs.\ 0.333/0.231/0.111) and faster inference (0.057\,s vs.\ 0.192\,s), indicating that explicit trend/seasonal decomposition helps but does not replace the richer long–short range interactions achieved by our state-space mixing.

\paragraph{Throughput and speedups} Beyond raw latency, we report \emph{MS\textsuperscript{3}M Speedup (×)} $=(\text{Latency of method})/(\text{Latency of MS\textsuperscript{3}M})$, so larger values favor MS\textsuperscript{3}M. In our setting, MS\textsuperscript{3}M is \textbf{10.28$\times$} faster than Crossformer (0.586\,s vs.\ 0.057\,s), \textbf{7.28$\times$} faster than FEDformer (0.415\,s), \textbf{5.23$\times$} faster than Informer (0.298\,s), \textbf{4.02$\times$} faster than TFT (0.229\,s), \textbf{4.09$\times$} faster than PatchTST (0.233\,s), \textbf{3.81$\times$} faster than iTransformer (0.217\,s), and \textbf{3.37$\times$} faster than ETSformer (0.192\,s). These speedups, together with fewer parameters than most competitors, translate into higher forecast rates per device and lower per-inference cost—key for real-time, edge, and large-scale deployment—reinforcing MS\textsuperscript{3}M’s position on the accuracy–efficiency Pareto frontier.\\

\noindent\fbox{%
  \parbox{\columnwidth}{%
    \centering
    \(\displaystyle
      \text{ETSformer} \lesssim \text{iTransformer} \lesssim \text{TFT} \lesssim \text{PatchTST} \lesssim \text{Informer} \lesssim \text{FEDformer} \lesssim \text{Crossformer}
    \)\\[3pt]
    \textit{(Ordered by MS\textsuperscript{3}M Speedup $\times$, lower $\to$ higher: 3.37, 3.81, 4.02, 4.09, 5.23, 7.28, 10.28)}
  }%
}
\\

\paragraph{Why MS\textsuperscript{3}M works}
We attribute the gains to three factors. 
\emph{(i) Multi-scale state mixing} captures local and long-range dependencies with favorable scaling, $\mathcal{O}(L\,d\,d_s)$, avoiding the quadratic dependence on $d$ common in stacked attention. 
\emph{(ii) Task-aligned conditioning} integrates exogenous features via a compact latent state ($F{=}13$, $d_{\mathrm{state}}{=}64$), improving tail stability without future leakage. 
\emph{(iii) Pareto efficiency:} except for Crossformer (which trades small error gains for $\sim$10$\times$ higher latency and $>2\times$ parameters), MS\textsuperscript{3}M simultaneously improves accuracy while reducing runtime and size, making it well-suited for leakage-safe, real-time forecasting.\\

\noindent\fbox{%
  \parbox{\columnwidth}{%
    \centering
    Accuracy–Efficiency Ranking (best $\to$ worse)\\[2pt]
    \(\displaystyle
      \textbf{MS}^{3}\textbf{M}
      \;\lesssim\;
      \text{ETSformer}
      \;\lesssim\;
      \text{Crossformer}
      \;\lesssim\;
      \text{Informer} \approx \text{FEDformer}
      \;\lesssim\;
      \text{TFT}
      \;\lesssim\;
      \text{PatchTST} \approx \text{iTransformer}
    \)\\[3pt]
    \scriptsize{Jointly considers lower errors, lower latency, and fewer parameters (per Table~\ref{tab:master_perf_complex}).}
  }%
}
\\

\paragraph{Performance and diagnostics}
On the held-out test set, the model achieves low absolute and squared error (RMSE $=0.290$~dBm; MAE $=0.169$~dBm) and very high explained variance ($R^2=0.9931$), with narrow 95\% bootstrap confidence intervals (Table~\ref{tab:s6_metrics}). 
Figure~\ref{fig:s6_rsrp_compact}a shows that predictions follow ground truth closely over the last 1000 samples, while the parity plot (Fig.~\ref{fig:s6_rsrp_compact}b) concentrates near the identity line. 
Residual analyses indicate near-zero bias and approximate normality (Fig.~\ref{fig:s6_rsrp_compact}c--d), no strong heteroscedasticity across the predicted range (Fig.~\ref{fig:s6_rsrp_compact}e), and limited temporal autocorrelation remaining in residuals (Fig.~\ref{fig:s6_rsrp_compact}h), suggesting that the model captures most of the predictable structure in the series. 
The $|$error$|$ CDF and boxplot (Fig.~\ref{fig:s6_rsrp_compact}f--g) further confirm that typical errors are small. 
Permutation importance (Fig.~\ref{fig:s6_rsrp_compact}i) highlights radio link quality and scheduling/context variables---notably RSRP, RSSI, SINR, PMI, and CQI---as primary drivers; additional contributions arise from RSRQ and throughput/coding indicators (SE, RI, MCS, BLER), as well as PRB allocation and traffic state (delay, buffer). 
Together, these diagnostics support the reliability and robustness of the forecaster on the test distribution.

% ====== Figure: Holistic diagnostics & importance ======
\begin{figure*}[t]
  \centering
  \includegraphics[width=\textwidth]{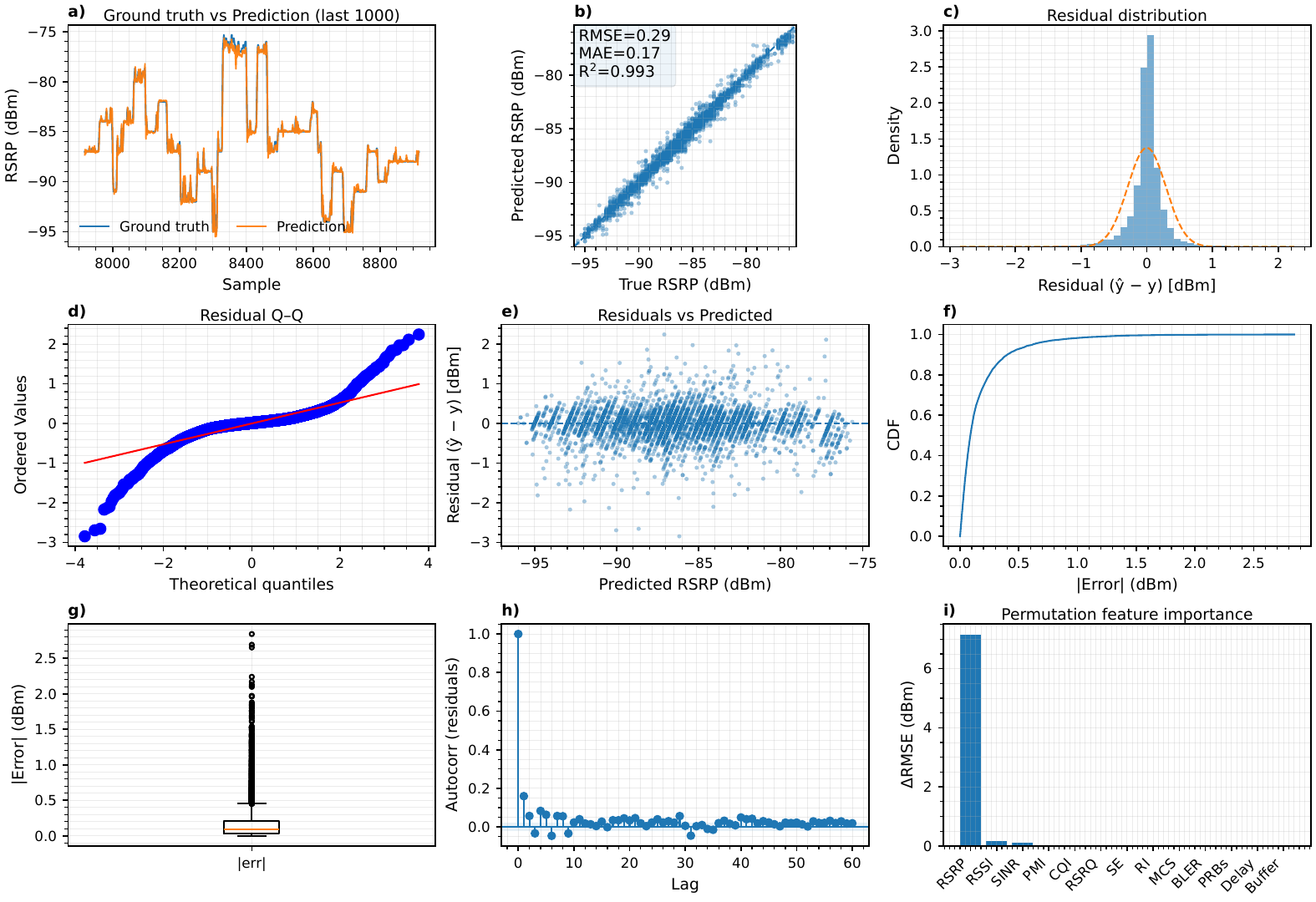}
  \caption{Comprehensive test-set diagnostics for the MS\textsuperscript{3}M RSRP forecaster. 
  \textbf{(a)} Ground truth vs.\ prediction for the last 1000 samples: the prediction closely tracks the measured RSRP with minimal phase lag and small amplitude error, illustrating stable short-horizon behavior on recent data. 
  \textbf{(b)} Parity plot ($\hat{y}$ vs.\ $y$): points cluster tightly around the identity line, consistent with low error (annotated RMSE and MAE) and high explained variance ($R^2\approx 0.993$). 
  \textbf{(c)} Residual distribution: histogram is narrow and approximately Gaussian, centered near zero, indicating low bias and a concentrated error profile in dBm. 
  \textbf{(d)} Residual Q--Q: empirical quantiles align well with the theoretical normal line; only mild tail deviations are visible, suggesting near-normal residuals. 
  \textbf{(e)} Residuals vs.\ predicted: no strong trend or funnel shape, indicating no pronounced heteroscedasticity across the predicted range. 
  \textbf{(f)} $|$Error$|$ CDF: the curve rises steeply, showing that a large fraction of samples have small absolute error (sub-dBm to low-dBm range), consistent with precise predictions. 
  \textbf{(g)} $|$Error$|$ boxplot: a compact interquartile range and low median reaffirm that typical errors are small. 
  \textbf{(h)} Residual autocorrelation: most lags lie within the (approx.) 95\% Bartlett band, indicating little remaining temporal structure in residuals (i.e., limited leftover predictability). 
  \textbf{(i)} Permutation feature importance ($\Delta$RMSE in dBm): increases in RMSE after feature-wise permutation quantify sensitivity. Radio-quality indicators such as \textit{RSRP}, \textit{RSSI}, \textit{SINR}, \textit{PMI}, and \textit{CQI} emerge among the most influential, followed by \textit{RSRQ}, spectral-efficiency/coding/throughput descriptors (e.g., \textit{SE}, \textit{RI}, \textit{MCS}, \textit{BLER}), and scheduler/traffic context (e.g., \textit{PRBs}, delay, buffer occupancy). These attributions are measured directly on the test set in original dBm units. \label{fig:s6_rsrp_compact}}
\end{figure*}

\section{Conclusion and Future Work}
This paper introduced \emph{MS$^{3}$M}, a lightweight multi-scale structured state-space mixture for leakage-safe, near-real-time KPI forecasting in agentic 6G O\!-RAN. By mixing HiPPO–LegS kernels across learned time scales, discretized via Tustin to ensure Schur stability, and combining them with squeeze–excitation gating and a compact GLU mixer, MS$^{3}$M achieves Transformer-competitive accuracy with substantially lower latency and footprint. On our O\!-RAN testbed dataset (13 KPIs), MS$^{3}$M delivers strong next-step RSRP performance while offering $3$–$10\times$ lower inference latency than representative Transformer baselines under a unified, leakage-safe protocol. The resulting accuracy–efficiency trade-off makes MS$^{3}$M a practical fit for Near-RT RIC xApps that require fast, reliable predictions to enable anticipatory control.

\paragraph*{Limitations}
Despite these results, several limitations should be acknowledged. (i) \textbf{Dataset scope:} evaluations are conducted on a bespoke testbed with a fixed KPI set and operating modes; generalization to other vendors, frequency bands, carrier bandwidths, or mobility regimes (e.g., dense handovers) remains to be validated. (ii) \textbf{Forecasting horizon:} the current study focuses on strict one-step-ahead prediction; many O\!-RAN decisions benefit from multi-horizon trajectories and temporal quantification of risk. (iii) \textbf{Uncertainty and robustness:} the reported model is point-predictive; principled uncertainty quantification, calibration, and robustness to concept drift, outliers, and extended missingness have not been exhaustively characterized. (iv) \textbf{Closed-loop impact:} we assess open-loop forecasting accuracy and latency; end-to-end effects on closed-loop RIC policies (xApps/rApps interacting via A1/E2/O1) and negotiation among agents under prediction errors are not measured here. (v) \textbf{Hardware/implementation:} while compact, our implementation targets general-purpose hardware; co-design with accelerators, quantization, and memory-aware kernels may change relative latencies across baselines. (vi) \textbf{Model scope:} MS$^{3}$M is KPI-agnostic but channel-independent in its depthwise SSM core; stronger cross-UE or topology-aware interactions (e.g., inter-cell interference, mobility graphs) are not explicitly modeled.

\paragraph*{Future Studies}
We see several promising directions: 
\begin{itemize}[leftmargin=*,nosep]
  \item \textbf{Multi-horizon \& probabilistic forecasting:} extend MS$^{3}$M with distributional heads (e.g., mixture/log-likelihood training), conformal prediction for finite-sample coverage, and trajectory rollouts for mid/long horizons.
  \item \textbf{Online adaptation \& drift handling:} incorporate change-point detection, test-time adaptation, and continual learning to track dynamics under evolving traffic, interference, or configuration updates.
  \item \textbf{Hybrid mixers:} explore SSM–attention hybrids (selective cross-channel/self-attention atop SSM backbones) and per-channel/time adaptive step sizes to better capture abrupt regime shifts.
  \item \textbf{Graph- and physics-aware structure:} inject cross-UE/cell relations (e.g., handover graphs, PRB contention) and lightweight domain constraints to improve extrapolation and interpretability.
  \item \textbf{Federated/edge training:} evaluate privacy-preserving learning across distributed RAN sites with heterogeneous data and bandwidth constraints.
  \item \textbf{Closed-loop RIC evaluation:} integrate forecasts into real xApps (e.g., link adaptation, PRB scheduling, mobility robustness) and quantify end-to-end gains, stability margins, and negotiation outcomes under uncertainty.
  \item \textbf{Broader benchmarks:} validate on multi-vendor datasets and public time-series suites, with standardized leakage-safe splits and ablations over kernel length, state size, and mixture count.
\end{itemize}
In summary, MS$^{3}$M advances the feasibility of control-grade forecasting in Near-RT RICs by pairing stability-aware state-space mixing with modern gating and compact channel mixing. We hope the released code and pipeline will catalyze rigorous, leakage-safe comparisons and accelerate the deployment of prediction-aware, agentic O\!-RAN.

%\section*{Acknowledgments}
%Omitted for double\mbox{-}blind review.

\bibliographystyle{IEEEtran}
\bibliography{main}

\end{document}